\documentclass[11pt]{article}

\usepackage[T1]{fontenc}
\usepackage{lmodern}
\usepackage{microtype}
\usepackage{fullpage}
\usepackage{authblk}
\usepackage{amsthm}
\usepackage{amssymb}
\usepackage[english]{babel}
\usepackage[utf8]{inputenc}
\usepackage{amsmath}
\usepackage{amsfonts}
\usepackage{graphicx}
\usepackage{mathtools}
\usepackage{xspace}
\usepackage{bm}

\usepackage{color}
\definecolor{darkgreen}{rgb}{0,0.5,0}
\definecolor{darkblue}{rgb}{0,0,0.8}
\definecolor{darkred}{rgb}{0.8,0,0}

\usepackage{makecell}
\usepackage{tikz}
\usetikzlibrary{shapes,arrows,positioning}

\usepackage{hyperref}
\hypersetup{
   unicode=false,          
   colorlinks=true,        
   linkcolor=darkred,          
   citecolor=darkblue,        
   filecolor=magenta,      
   urlcolor=brown           
}

\newtheorem{definition}{Definition}[section]
\newtheorem{lemma}[definition]{Lemma}
\newtheorem{theorem}[definition]{Theorem}

\newtheorem{corollary}[definition]{Corollary}

\usepackage[noend, noline, ruled, linesnumbered]{algorithm2e}
\SetAlFnt{\normalsize}
\DontPrintSemicolon
\SetKwComment{Comment}{$\triangleright$\ }{}
\SetKwProg{Def}{def}{:}{}

\SetKwFunction{KwScore}{score}
\SetKwFunction{KwApproximate}{approximate}
\SetKwFunction{KwWeighted}{generalized\_weighted\_matching}
\SetKwData{KwFalse}{false}
\SetKwData{KwTrue}{true}
\graphicspath{{./figures/}}

\newcommand{\bigo}{\mathcal{O}}
\newcommand{\ignore}[1]{}

\newcommand{\sgn}{\mathrm{sgn}}

\newcommand{\score}{S}

\title{\bf Faster Approximate(d) Text-to-Pattern \bm{$L_1$} Distance}

\author{\LARGE Przemys\l{}aw~Uzna\'nski}
\affil{
\Large Department of Computer Science,\\ ETH Zürich, Switzerland}
\date{}
\begin{document}
\maketitle

\thispagestyle{empty}

\abstract{The problem of finding \emph{distance} between \emph{pattern} of length $m$ and \emph{text} of length $n$ is a typical way of generalizing pattern matching to incorporate dissimilarity score. For both Hamming and $L_1$ distances only a super linear upper bound $\widetilde\bigo(n\sqrt{m})$ are known, which prompts the question of relaxing the problem: either by asking for $(1 \pm \varepsilon)$ approximate distance (every distance is reported up to a multiplicative factor), or $k$-approximated distance (distances exceeding $k$ are reported as $\infty$). We focus on $L_1$ distance, for which we show new algorithms achieving complexities respectively $\widetilde\bigo(\varepsilon^{-1} n)$ and $\widetilde\bigo((m+k\sqrt{m}) \cdot n/m)$. This is a significant improvement upon previous algorithms with runtime $\widetilde\bigo(\varepsilon^{-2} n)$ of Lipsky and Porat [Algorithmica 2011] and $\widetilde\bigo(n\sqrt{k})$ of Amir, Lipsky, Porat and Umanski [CPM 2005]. 
}

\section{Introduction}
One of the fundamental problems in text algorithms is, given text $T$ of length $n$ and pattern $P$ of length $m$, both over some (integer) alphabet $\Sigma$, the computation of distance between $P$ and every $m$-substring of $T$. This is a popular way of searching for imperfect occurrences of pattern in the text, generalizing standard pattern matching problem. Two distance functions received most attention in the context of integer sequences, Hamming distance and $L_1$ distance. 

Those problems exhibit the usual story (in pattern matching problems) of starting with a naive quadratic time upper bound, and the goal being to develop close to linear time one. The result of \cite{Abrahamson87} have shown algorithm computing text-to-pattern Hamming distance in time $\bigo(n \sqrt{m \log m})$. This algorithm goes beyond the usual repertoire of \emph{combinatorial} tools, and uses boolean convolution as a subroutine.
Using similar approach, \cite{DBLP:conf/cpm/CliffordCI05} and \cite{Amir2005} have shown identical upper bound for $L_1$ distance version of the problem. 

However, the bound of $\widetilde\bigo(n^{3/2})$ is unsatisfactory. What followed was a compelling argument (see note \cite{Clifford}) showing that any significant improvement to the bound for Hamming distances by a \emph{combinatorial} algorithm leads automatically to an improvement in the complexity of \emph{boolean matrix multiplication}, suggesting that further progress might be hard. Later, a direct reduction from Hamming distance problem to $L_1$ one was shown \cite{DBLP:journals/ipl/LipskyP08a}, and a reverse reduction was presented in \cite{DBLP:journals/corr/abs-1711-03887}, together with full suite of two way reductions between those two metrics and other score functions,  linking the complexity of those problems together (ignoring poly-logarithmic factors).

Thus, the natural next step in is to consider relaxations to those two problems, i.e. require only reporting of a multiplicative $(1\pm\varepsilon)$ approximation of the distances. A result of \cite{Karloff93} has shown how to use random $\Sigma \to \{0,1\}$ projections to achieve approximation of the Hamming distance in time $\bigo(\varepsilon^{-2} n \log^3 m)$. Many believed that the so-called \emph{variance bound}, that is $\varepsilon^{-2}$ dependency is tight, given the evidence of similar lower-bounds in sketching of Hamming distance (c.f. \cite{DBLP:conf/soda/Woodruff04}, \cite{DBLP:journals/toc/JayramKS08}, \cite{DBLP:journals/siamcomp/ChakrabartiR12}). However, in a breakthrough paper \cite{DBLP:conf/focs/KopelowitzP15} presented (a quite involved) algorithm working in time $\bigo(\varepsilon^{-1} n \log \varepsilon^{-1} \log n \log m \log |\Sigma|)$. That result was recently simplified in \cite{DBLP:conf/soda/KopelowitzP18}, with a slightly improved runtime $\bigo(\varepsilon^{-1} n \log n \log m)$. For $L_1$ distance, the only known approximation algorithm was one from \cite{DBLP:journals/algorithmica/LipskyP11}, having runtime $\bigo(\varepsilon^{-2} n \log m \log |\Sigma|)$. The question whether the barrier of $\varepsilon^{-2}$ can be beaten for $L_1$ distance remained open.

Another standard way of relaxing exact text-to-pattern distance is to report exactly only the values not exceeding certain threshold value $k$, the so-called $k$-approximated distance. The motivation for this comes from interpretation of exact text-to-pattern Hamming distance as simply counting \emph{mismatches} in exact pattern matching, and then $k$-approximated Hamming distance becomes reporting only alignments where there are at most $k$ mismatches. The very first solution to the Hamming distances version of this problem was shown in \cite{LandauV86} working in time $\bigo(nk)$, using essentially a very combinatorial approach of taking $\bigo(1)$ time per mismatch per alignment using LCP queries. This initiated a series of improvements to the complexity, with algorithms of complexity $\bigo(n \sqrt{k \log k})$ and $\bigo((k^3 \log k + m)\cdot n/m)$ in \cite{AmirLP04}, later improved to $\bigo((k^2 \log k + m\ \text{poly} \log m)\cdot n/m)$ by \cite{CliffordFPSS16} and finally $\bigo( (m \log^2 m \log |\Sigma| + k \sqrt{m \log m})\cdot n/m)$ by \cite{DBLP:journals/corr/GawrychowskiU17}. The last result also provides an argument that the $\bigo((m + k\sqrt{m}) \cdot n/m)$ might be tight up to sub-polynomial factors, by extending argumentation from \cite{Clifford} to incorporate $k$ into the reduction. On the $L_1$ side, the only known $k$-approximated algorithm was from \cite{Amir2005} with complexity of $\bigo(n \sqrt{k \log k})$. A question on whether one can design algorithms that for polynomially large values of $k$ still work in almost-linear time remained open (such as it for the Hamming distance when $k \le \sqrt{m}$).

\paragraph{Problem definition and preliminaries.}
Let $X = x_1 x_2 \ldots x_n$ and $Y = y_1 y_2 \ldots y_n$ be two words over integer alphabet $[M]$ for some constant $M = \text{poly}(n)$. We define their $L_1$ distance as $L_1(X,Y) = \sum_i |x_i - y_i|$, and their Hamming distance as $\text{Ham}(X,Y) = |\{ i : x_i \not= y_i \}|$. 

The \emph{exact} $L_1$ text-to-pattern distance between text $T = t_1 t_2 \ldots t_{n}$ and pattern $P = p_1 p_2 \ldots p_{m}$ is defined as an array $\score$ such that $\score[i] = L_1(T[i+1\ ..\ i+m], P) = \sum_{j=1}^{m} |t_{i+j} - p_{j}|$. The $(1\pm\varepsilon)$ \emph{approximate} text-to-pattern $L_1$ distance is the array $\score_\varepsilon$ such that for all $i$, $(1-\varepsilon) \cdot \score[i] \le \score_{\varepsilon}[i] \le (1+\varepsilon) \cdot \score[i]$.
The $k$-\emph{approximated} text-to-pattern $L_1$ distance is the array $\score_k$ such that for all $i$, $\score_k[i] = \score[i]$ when $\score[i] \le k$ and $\score_k[i] = \infty$ when $\score[i] > k$. The definitions for exact, approximate and approximated Hamming distance follow in the same manner.

We assume that all of the values in the input are positive. If not, then we can add some large integer $N$ to every value of input without changing the $L_1$ distance. Let $M = \text{poly}(n)$ be the upperbound on every value of the input. We also assume RAM model, with words big enough to hold integers up to $M$, and having arithmetic operations over those in constant time. Unless stated otherwise, we denote text length as $n$ and pattern length as $m$.

We define the size of \emph{run length encoding} (RLE) of a string as a number of different runs (maximal sequences of identical letters). We say that string is $k$-RLE if its RLE is at most $k$.

\paragraph{Our results.}
We show improved algorithms for both $(1\pm\varepsilon)$ approximation and $k$-approximated version of the text-to-pattern $L_1$ distance. 
\begin{theorem}
\label{th:approximate}
There is a randomized Monte Carlo algorithm that outputs $(1\pm\varepsilon)$ approximation of $L_1$ text-to-pattern distance in time $\bigo(\varepsilon^{-1}n \log^3 n \log m)$. The algorithm works with high probability.
\end{theorem}

\begin{theorem}
\label{th:approximated}
There is a deterministic algorithm that outputs $k$-approximated $L_1$ text-to-pattern distance in time $\bigo((m \log^3 m + m \log^2 n + k \sqrt{m \log m} \cdot \log^2 n) \cdot n/m)$. 
\end{theorem}
This shows, that similarly to reporting $k$-approximated Hamming distance, one can report all positions exactly where the $L_1$ distance is at most $\sqrt{m}$ in almost linear time. 



Our results show that in the text-to-pattern approximate/approximated distance reporting, there does not seem to be significant difference in the complexities of Hamming and $L_1$ distance versions (up to current upper bounds). 

We also link $k$-approximated Hamming and $L_1$ distance problems.
\begin{theorem}
\label{th:ham_l1_reduction}
Let $T(n,m,k)$ be the runtime of $k$-approximated text-to-pattern Hamming distance. Then $k$-approximated text-to-pattern $L_1$ distance is computed in time $\bigo(n \log^3m + T(m,m,k)\cdot \log^2 n \cdot n/m)$
\end{theorem}

\begin{corollary}
\label{th:reduction2}
If the $k$-approximated text-to-pattern  Hamming distance is computed in time $\widetilde\bigo(n+(k \sqrt{m})^{1-\delta} \cdot n/m)$ for $\delta\ge0$ then $k$-approximated text-to-pattern  $L_1$ distance is computed in time $\widetilde\bigo(n+(k \sqrt{m})^{1-\delta} \cdot n/m)$ as well.
\end{corollary}

\begin{figure}[t]
\tikzstyle{lineDotted} = [draw, -latex', thick, dotted]
\tikzstyle{lineDashed} = [draw, -latex', thick, dashed]
\tikzstyle{lineTwoDashed} = [draw, latex'-latex', thick, dashed]

\tikzstyle{lineSolid} = [draw, -latex', thick]
\tikzstyle{block} = [draw, rectangle, text centered, node distance = 6em]
\centering
\begin{tikzpicture}

\node[block, anchor=center] (approximatedHam){\makecell[c]{$k$-approximated\\ Hamming distance}};
\node[block, right=6em of approximatedHam] (approximatedL1){\makecell[c]{$k$-approximated\\ $L_1$ distance}};
\node[block, left=6em of approximatedHam] (BMM){\makecell[c]{Boolean Matrix\\Multiplication}};
\node[block, below=6em of approximatedHam] (boundedRLEHam){\makecell[c]{$k$-RLE\\ Hamming distance}};
\node[block, below=6em of approximatedL1] (boundedRLEL1){\makecell[c]{$k$-RLE\\ $L_1$ distance}};

\path[lineDashed] ([xshift=-1em]approximatedHam.south) -- node[xshift=-2em]{\makecell[c]{\cite{DBLP:journals/corr/GawrychowskiU17}\\\eqref{cor:approx_to_bounded_rle_ham}}}([xshift=-1em]boundedRLEHam.north);
\path[lineSolid] ([xshift=1em]boundedRLEHam.north) -- node[xshift=1.5em]{\eqref{lem:rle_to_approximated}}([xshift=1em]approximatedHam.south);
\path[lineTwoDashed] (boundedRLEHam) -- node[yshift=1.5em]{\makecell[c]{\cite{DBLP:journals/corr/abs-1711-03887}\\\eqref{cor:reduction}}}(boundedRLEL1);
\path[lineDashed] (BMM) -- node[yshift=1em]{\cite{DBLP:journals/corr/GawrychowskiU17}} (approximatedHam);
\path[lineSolid] (approximatedL1.south) -- node[xshift=2em]{\eqref{th:approx_to_bounded_rle_l1}}(boundedRLEL1.north);
\end{tikzpicture}
\caption{Existing (dashed lines) and new (solid lines) reductions.}
\end{figure}
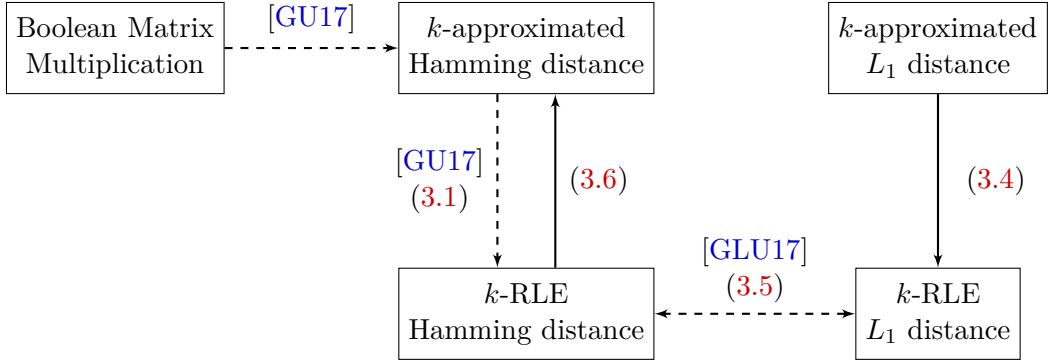


\paragraph{Overview of the techniques.}
Main technique used in our approximation algorithm, just as in the previous work of \cite{DBLP:journals/algorithmica/LipskyP11}, are the \emph{generalized weighted mismatches}: given arbitrary weight function $\sigma : \Sigma \times \Sigma \to \mathbb{Z}$, we output array $\score_{\sigma}$ such that $\score[i] = \sum_{j=1}^m \sigma(t_{i+j}, p_j)$. We use the following algorithm described first by \cite{DBLP:journals/algorithmica/LipskyP11}, that computes $\score_{\sigma}$ in time $\bigo(|\Sigma| n \log m)$: for each $c \in |\Sigma|$, the contribution of this letter can be computed from convolution of two vectors,  first being $\chi_c(P)$, the characteristic vector of letter $c$ in $P$, and second being $\{\sigma(c,t_i) \}_{i=1}^{n}$.

In approximated algorithm, we use the techniques of \emph{alignment filtering} and \emph{kernelization}, used in the context of Hamming distances by \cite{CliffordFPSS16}  and then refined and simplified by \cite{DBLP:journals/corr/GawrychowskiU17}. The general idea is to first consider the periodic structure of the pattern. If the pattern is not periodic enough, then $m$-substrings of text that have small distance to pattern must occur not too often. One can use approximate algorithm to filter out all the alignments with too large distance, and manually verify all the $\bigo(n/k)$ alignments that remain in time $\bigo(k)$ per each. If the pattern is periodic enough, then both the pattern and text can be rearranged into new instance, that retains letter alignments, and both new pattern and new text are compressible (both have RLE of only $\bigo(k)$ blocks). This reduces problem of lets say, text-to-pattern Hamming distances to the same problem, but with additional constraints on RLE of pattern and text.

\emph{Linearity preserving reductions} introduced in \cite{DBLP:journals/corr/abs-1711-03887}  are a formalization of existing previously reductions between metrics (cf. \cite{DBLP:journals/ipl/LipskyP08a}).
Main idea is that in order to show a reduction between two pattern-matching problems, one can represent them as lets say $(+,\diamond)$ and $(+,\square)$ convolutions, and show a reduction just between $\diamond$ and $\square$ binary operators. To make such reduction work, it needs to be of a specific form.
More precisely, fix integer $t$ being the size of reduction, integer coefficients $\alpha_1,\ldots,\alpha_t$ and functions $f_1, \ldots, f_t, g_1, \ldots, g_t$ such that for any $x,y$:
$$x\ \square\ y= \sum_{i=0}^t \alpha_i \cdot (f_i(x) \diamond g_i(y)).$$
Then $(+,\square)$-convolution of $T$ and $P$ is computed as a linear combination of $(+,\diamond)$-convolutions $f_i(P)$ with $g_i(P)$, where $f(X) = f(x_1)f(x_2)\ldots f(x_n)$ for $X = x_1x_2\ldots x_n$.
\section{Proof of Theorem~\ref{th:approximate}}

\begin{algorithm}[h]
\caption{$(1\pm\varepsilon)$-approximation of text-to-pattern $L_1$ distance.}
\label{alg:approximate}
\KwIn{Integer strings $T$ and $P$.}
\KwOut{Score vector $\score_\varepsilon$.}
\Def{$\KwScore(x,y)$}
{
	$x_0 \gets x \bmod 2$\;
	$y_0 \gets y \bmod 2$\;
	\uIf{$x_0 = y_0$}
	{
 		\Return{0}\;
	}
	\uElseIf{ $\sgn(x - y) = \sgn(x_0 - y_0)$ }
	{
 		\Return{$1$}\;
	}
	\uElse
	{
		\Return{$-1$}\;
	}
}
\;
\Def{$\KwApproximate(T,P)$}
{
	$\Delta \gets \text{ u.a.r. integer from } 0 \text{ to } 2^{\lceil \log M \rceil}-1$\;
	$T' \gets T+\Delta$\;
	$P' \gets P + \Delta$\;
	$\score_{\varepsilon} \gets [0 \ldots 0]$\;
	\For{$i \gets 0$ \KwTo $\lceil \log M \rceil$}
	{
		$T'' \gets \lfloor T'/2^i \rfloor \bmod 2^{b}$\;
		$P'' \gets \lfloor P'/2^i \rfloor \bmod 2^{b}$\;
		$S \gets \KwWeighted(T'',P'',\KwScore)$\;
		$\score_{\varepsilon} \gets \score_{\varepsilon} + S \cdot 2^i$\;
	}
	\Return{$\score_{\varepsilon}$}\;	
}
\end{algorithm}

In this section we prove Theorem~\ref{th:approximate}.
We use a procedure $\KwWeighted(T,P,\KwScore)$ that computes,
for a text $T= t_{1} t_{2}\ldots t_{n}$ and a pattern $P= p_{1}p_{2}\ldots p_{m}$
and an arbitrary weight function $\sigma : \Sigma \times \Sigma \to \mathbb{Z}$,
the array $\score_{\sigma}$ such that $\score[i] = \sum_{j=1}^m \sigma(t_{i+j}, p_j)$
in in $\bigo(|\Sigma| n \log m)$ time.

Let $\delta = \frac{\varepsilon}{24 \cdot (3+\log M)} = \Theta(\varepsilon/\log n)$, and let $b$ be the smallest positive integer such that $2^b \ge 1/\delta$.  We claim that with such parameters, Algorithm~\ref{alg:approximate}
outputs the desired $(1\pm\varepsilon)$-approximation in the claimed time.
Let $\score_{\varepsilon}$ be its output.

\begin{theorem}
For any $i$, $\score[i]\cdot (1-\varepsilon) \le \score_{\varepsilon}[i] \le \score[i]\cdot(1+\varepsilon)$ with probability at least $2/3$.
\end{theorem}
\begin{proof}
Consider first $x = x_a$ and $y = y_b$, two characters of the input. We analyze how well Algorithm~\ref{alg:approximate} approximates $|x-y| = \sgn(x-y) \cdot (x-y)$ in the consecutive calls of $\KwWeighted$.
First, fix value of $\Delta$ and consider the binary representations of $x' = x+\Delta$ and $y' = y+\Delta$. More precisely, let $x' = \sum_i 2^i \cdot \alpha_i$ and $y' = \sum_i 2^i \cdot \beta_i$ for some $\alpha_i,\beta_i \in \{0,1\}$. Algorithm~\ref{alg:approximate} in essence estimates $|x - y| = \sgn(x-y) \sum_i 2^i (\alpha_i - \beta_i)$ with $C = \sum_i 2^i \gamma_i$ where $\gamma_i \in \{-1,0,1\}$ is the estimation of a contribution of $(\alpha_i - \beta_i)\cdot \sgn(x-y)$ to $(x'-y')$ and depends only on values of $\alpha_j - \beta_j \in \{-1,0,1\}$ for $i \le j < i+b$ in the following way:
\begin{itemize}
\item If, for every $i \le j < i+b$ we have $\alpha_j = \beta_j$, then $\gamma_i = 0$.
\item Otherwise, let $j'$ be the largest $j$ such that $i \le j < i+b$ and $\alpha_j \not= \beta_j$. If $\alpha_{j'} - \beta_{j'} = 1$, then the local estimation is that $x' > y'$ and so $\gamma_i = \alpha_i - \beta_i$, and otherwise $\gamma_i = -1 \cdot (\alpha_i - \beta_i)$.
\end{itemize}

Consider $c = \max\{i : \alpha_i \not= \beta_i\}$ and $d = \max\{i : 2^i \le (x'-y')\}$, that is $c$ is the position of the highest bit on which $x'$ and $y'$ differ, and $d$ is the position of the highest bit of $x'-y'$. In general, $c \ge d$, and we say that pair $x',y'$ is $t$-\emph{bad}, if $c-d=t$.

We first observe that for a $x,y$ pair to be at least $t$-bad, a following condition must be met: $\lfloor x'/2^{d+t} \rfloor \not= \lfloor y'/2^{d+t} \rfloor$. Since $\Delta$ is chosen uniformly at random from a large enough range of integers, there is 
$$\sum_{\tau \ge t} \Pr(x',y' \text{ is }\tau\text{-bad }\big|\ x,y) \le |x-y|/2^{d+t} \le 2^{-t+1}.$$
 We also observe following:  for any pair $x',y'$, in $C$, all the coefficients $\gamma_c, \gamma_{c-1}, \ldots, \gamma_{c-b+1}$ are computed correctly, since for any $j$ such that $c \ge j \ge c-b+1$ there is $j' = c$, and then $\gamma_j = (\alpha_j - \beta_j) \cdot \sgn(x-y)$. Therefore 
$$\big|C - \left|x'-y'\right| \big| = \left|\sum_{i \le c-b} 2^i (\gamma_i - (\alpha_i - \beta_i)\cdot\sgn(x-y)) \right| \le  \sum_{i \le c-b} 2 \cdot 2^i < 2 \cdot 2^{c-b+1}.$$
If a pair $x',y'$ is $t$-bad, it immediately follows that the absolute error of estimation is at most $2^{c-b+2} = 2^{d+t-b+2} \le |x'-y'| 2^{t+2} \delta $.

We now estimate expected error in estimation based on choice of $\Delta$. If a particular pair $x,y$ is $t$-bad, then $t \le 1+\lceil \log M \rceil$. Using the previous observations, we have
\begin{align*}
\mathbb{E}\Big[  \big|C - \left|x'-y'\right|\big|\ \ \Big|\ \  x,y\Big] &= \sum_{t} \Pr(x',y' \text{ is }t\text{-bad }\big|\ x,y) \cdot  \mathbb{E}\Big[\big|C - \left|x'-y'\right| \big|\ \Big|\ x',y' \text{ is }t\text{-bad } \Big] \\ 
&\le \sum_{t=0}^{1+\lceil \log M \rceil} 2^{-t+1} |x-y|  2^{t+2} \delta = (3+ \log M ) 8 \delta |x-y| = \frac{\varepsilon}{3} |x-y|.
\end{align*}

By linearity of expectation $\mathbb{E}\Big[ \big|\score_\varepsilon[i] - \score[i]\big| \Big] \le \frac{\varepsilon}{3} \score[i],$
and by Markov's inequality the claim follows.
\end{proof}

Now, a standard amplification technique applies: it is enough to repeat Algorithm~\ref{alg:approximate} independently $p$ times and take the median value from $\score_{\varepsilon}^{(1)}[i], \score_{\varepsilon}^{(2)}[i], \ldots, \score_{\varepsilon}^{(p)}[i]$ as the final estimate $\widehat{\score}_{\varepsilon}[i]$. Taking $p = \Theta(\log n)$ to be large enough makes the final estimate good with high probability, and by the union bound whole $\widehat{\score}_\varepsilon$ is a good estimate of $\score$.

The complexity of Algorithm~\ref{alg:approximate} is dominated by $\KwWeighted$ being invoked $\bigo(\log n)$ times on alphabet of size $2^b = \Theta(\varepsilon^{-1} \log n)$. Each such invocation takes $\bigo(2^b n \log m) = \bigo(\varepsilon^{-1} n \log n \log m)$, and Algorithm~\ref{alg:approximate} takes $\bigo(\varepsilon^{-1} n \log^2 n \log m)$ time and the total time for computing $(1\pm\epsilon)$-approximation is $\bigo(\varepsilon^{-1} n \log^3 n \log m)$.

\section{Proofs of Theorem~\ref{th:approximated} and Theorem~\ref{th:ham_l1_reduction}}
In our construction of algorithm for $k$-approximated text-to-pattern $L_1$ distance, we make extensive use of techniques used in \cite{DBLP:journals/corr/GawrychowskiU17} for solving $k$-approximated Hamming distances.
More precisely, we need two components:

\begin{corollary}[\cite{CliffordFPSS16},\cite{DBLP:journals/corr/GawrychowskiU17}]
\label{cor:approx_to_bounded_rle_ham}
The $k$-approximated text-to-pattern Hamming distance problem reduces in $\widetilde\bigo(n)$ time to $\bigo(n/m)$ instances of  text-to-pattern  Hamming distance on $\bigo(k)$-RLE inputs of length $\bigo(m)$.
\end{corollary}

\begin{corollary}[\cite{DBLP:journals/corr/GawrychowskiU17}]
\label{cor:run_complexity}
Text-to-pattern Hamming distance on $k$-RLE inputs of length $\bigo(m)$ is computed exactly in time $\bigo(m+k \sqrt{m \log m})$ time.
\end{corollary}

We also need a following definition, as in \cite{DBLP:journals/corr/GawrychowskiU17} and \cite{CliffordFPSS16}, that an integer $\pi>0$ to is a $x$-period
of a string $S[1,m]$, if $\text{Ham}(S[\pi+1,m],S[1,m-\pi]) \le x$.

\begin{lemma}[Fact 3.1 in \cite{CliffordFPSS16}]
\label{lem:filtration}
If the minimal $2x$-period of the pattern is $\ell$, then for any two distinct $m$-substrings of text with Hamming distance to pattern at most $x$, their starting positions are at distance at least $\ell$.
\end{lemma}

We start with a $L_1$ version of Corollary~\ref{cor:approx_to_bounded_rle_ham}.

\begin{theorem}
\label{th:approx_to_bounded_rle_l1}
The $k$-approximated text-to-pattern $L_1$ distance problem reduces in $\bigo(n \log^3m)$ time to $\bigo(n/m)$ instances of exact $L_1$ text-to-pattern $L_1$ distance with $\bigo(k)$-RLE inputs of length $\bigo(m)$, where both pattern and text might have wildcards.
\end{theorem}

\begin{proof}
By a standard trick, it is enough to consider case where $T$ is of length $2m$, as any other case can be reduced to $\lceil n/m \rceil$ instances of this type. We proceed by showing key features of reduction in Corollary~\ref{cor:approx_to_bounded_rle_ham}.

Observe that for any words over integer alphabet $X,Y$, there is $L_1(X,Y) \ge \text{Ham}(X,Y)$. This makes any technique eliminating alignments with too large Hamming distance correct for filtering $L_1$ distances as well. 
We can determine easily minimal $\bigo(k)$-period of the pattern. As in \cite{DBLP:journals/corr/GawrychowskiU17}, we run Karloff's approximate Hamming distances algorithm \cite{Karloff93} matching pattern against pattern, with precision $1+\varepsilon=2$. This takes $\bigo(m \log^3 m)$ time, and we end up with one of two cases:
\begin{itemize}
\item every $4k$-period of the pattern is at least $k$, or
\item there is  a $8k$-period of the pattern that is at most $k$.
\end{itemize}

\paragraph{No small $4k$-period.}
We run Karloff's algorithm on pattern against the text, with $1+\varepsilon=2$. We then filter out all alignments where there were more than $2k$ reported mismatches. By Lemma~\ref{lem:filtration}, there are $\bigo(m/k)$ such alignments. By the relation between $L_1$ and Hamming distances, all discarded alignments were safe to do so for $L_1$ distances as well. Then, we test every such position using the ``kangaroo jumps''
technique of Landau and Vishkin~\cite{LandauV86}, using $\bigo(k)$ constant-time operations per position,
in total $\bigo(m)$ time. The only modification in this part from \cite{DBLP:journals/corr/GawrychowskiU17} approach is that with each found mismatch through, we account for the $L_1$ score it generates.

\paragraph{Small $8k$-period.}
Let $\ell$ be such a small $8k$-period. The initial approach from \cite{DBLP:journals/corr/GawrychowskiU17} in this case summarizes as follows. First, a subword $T'$ of $T$ is located, that contains all alignments of $P$ that match with Hamming distance at most $k$. Thus for our purposes it contains all alignments with $L_1$ distance to $P$ at most $k$ as well (c.f. Lemma 2.4 in \cite{DBLP:journals/corr/GawrychowskiU17}). Then both $P$ and $T'$ are padded with \emph{special characters}, and are subsequently rearranged into $P^\star$ and $T^\star$ of following properties (c.f. Lemma 2.5 in \cite{DBLP:journals/corr/GawrychowskiU17}):
\begin{itemize}
\item Both $P^\star$ and $T^\star$ are of $\bigo(m)$ length.
\item Both $P^\star$ and $T^\star$ have $\bigo(k)$ runs.
\item There is a map $i \to j$ such that if pair of letters is aligned between $P$ and $T'[i,i+|P|-1]$, then there is corresponding aligned pair of identical letters in $P^\star$ and $T^\star[j,j+|P^\star|-1]$. Moreover, any additional aligned pair of letters in $P^\star$ and $T^\star$ involve at least one special letter. The map depends only on values of $|P|$, $|T'|$ and $\ell$.
\end{itemize}

The last property, coupled with invariance of Hamming distance under permuting of input (as long as we preserve alignments) means that text-to-pattern Hamming distances between $P^\star$ and $T^\star$ encode, under constant additive term and reordering, text-to-pattern Hamming distances between $P$ and $T'$, and thus all small between $P$ and $T$. We use the same transformation and claim that if all special characters used were wildcards $\ast$ (symbols that have $L_1$ distance 0 to every other character), the $L_1$ distance is preserved.
\end{proof}

Now, instead of building explicit algorithm for computing $L_1$ distance on bounded RLE instances, we make use of existing algorithm for Hamming distances and a reduction that preserves bounds on RLE.

\begin{corollary}[c.f. Theorem 2.1, Theorem 2.2 and Lemma A.1 in \cite{DBLP:journals/corr/abs-1711-03887}]
\label{cor:reduction}
For any $M\ge0$, there is a linearity preserving reduction from $L_1$ distance between integers from $[M]$ to $\bigo(\log^2 M)$ instances of Hamming distance. There is a converse reduction from Hamming distance to $\bigo(1)$ instances of $L_1$ distance. Those reductions allows for wildcards in the input and produces wildcard-less instances on the output.
\end{corollary}

We now have enough tools to construct $k$-approximated $L_1$ distance algorithm of desired runtime. 
\begin{proof}[Proof of Theorem~\ref{th:approximated}]
By Theorem~\ref{th:approx_to_bounded_rle_l1}, we reduce the input instance to $\bigo(n/m)$ instances of bounded RLE $L_1$ distances. By fixing $M = \text{poly}(n)$ in Corollary~\ref{cor:reduction} each of those reduces to $\bigo(\log^2 n)$ instances of text-to-pattern Hamming distance. Additionally, the reduction does not create any new runs, so the output instances have the same bound on RLE, so the Corollary~\ref{cor:run_complexity} applies. Final runtime is $\bigo((m+ k \sqrt{m \log m}) \cdot \log^2 n + n \log^3m)$.
\end{proof}

To finish the full chain of reductions, we also show a following.
\begin{lemma}
\label{lem:rle_to_approximated}
Text-to-pattern Hamming distance on $k$-RLE inputs with text and pattern of length $\bigo(m)$ reduces to $\bigo(1)$ instances of $2k$-approximated Hamming distance on inputs of length $\bigo(m)$.
\end{lemma}
\begin{proof}
We proceed in a manner similar to the \cite{DBLP:journals/corr/GawrychowskiU17}. We observe that it is enough to compute the second discrete derivate of the output array $\score$, that is $D^2 \score$ defined as $(D^2 \score)[i] = \score[i+2]-2\score[i+1]+\score[i]$, since having $D^2 \score$ and two initial values of $\score$ (later can be compute naively in time $\bigo(m)$) is enough to recover all of $\score$. For any two blocks $t_ut_{u+1} \ldots t_{v-1}t_{v}$ and $p_yp_{y+1} \ldots p_{z-1}p_z$ of the same letter, $D^2 \score$ needs to be updated in only $4$ places, that is $D^2 \score[u-z] += 1$, $D^2 \score[v-z+1] -= 1$, $D^2 \score[u-y+1] -= 1$ and $D^2\score[v-y+2]$. We now explain how to deal with the first kind of updates, with the three following being done in an analogous manner.

We first reduce to a problem of \emph{$k$-sparse text-to-pattern Hamming distance}, where text and pattern are of length $\bigo(m)$ and have each at most $k$ regular characters, with every other character being wildcard $\ast$ (special character having 0 distance to every other character).
We construct sparse instance as follows: for every position $t_u$ in $T$ that starts a block, we set $T_{\text{sparse}}[u] = t_u$, and similarly in pattern for a position $p_y$ (that starts a block), we set $P_{\text{sparse}}[y] = p_y$. Observe, that if $t_u \not= p_y$, then in the answer there is $\score_{\text{sparse}}[u-y] +=1$, and if $t_u = p_y$ then $\score_{\text{sparse}}$ remains unchanged (answer counts mismatches, while we want to count matches). To invert the answer, we create $T_{\text{bin}}$ such that $T_{\text{bin}}[i] = 1$ iff $T_{\text{sparse}}[i] \not= \ast$ and $T_{\text{bin}}[i] = 0$ otherwise, and $P_{\text{bin}}$ in an analogous manner. Then a single convolution of $T_{\text{bin}}$ with $P_{\text{bin}}$ counts for every alignment total number of non-special text characters that were aligned with non-special pattern characters. A single subtraction yields answer.

To reduce from $k$-sparse instances of Hamming distance to $2k$-approximated Hamming distance, we follow an analogous reduction from \cite{DBLP:journals/corr/abs-1711-03887} (c.f. Lemma A.1) that reduces Hamming distance on $\mathbb{N}^+ + \{\ast\}$ to Hamming distance on $\mathbb{N}$. Write first instance such that $T_1[i] = T[i]$ iff $T[i] \not= \ast$ and $T_1[i] = 0$ iff $T[i] = \ast$ (with analogous transformation on $P$ to compute $P_1$). Write second instance such that $T_2[i]=1$ iff $T_2[i]\not=\ast$ and $T_2[i]=0$ iff $T[i] = \ast$ (with the same transformation on $P$ to compute $P_2$). Now we observe that $\text{Ham}(T[i],P[j]) = \text{Ham}(T_1[i],P_1[j]) - \text{Ham}(T_2[i],P_2[j])$, thus it is enough to compute exact Hamming text-to-pattern distances on those two instances and subtract them. However, we observe that in both of them, there are in total at most $2k$ characters different than 0, thus $2k$-approximated Hamming distance works just as fine.
\end{proof}

We now observe that the proof of Theorem~\ref{th:ham_l1_reduction} follows automatically from plugging Lemma~\ref{lem:rle_to_approximated} in place of Corollary~\ref{cor:run_complexity} in the proof of Theorem~\ref{th:approximated}.
We conclude this section with a remark on necessity of applying reduction to kernelized version of the problems, instead of directly to approximated problems. That is, some coefficients $\alpha_i$ are negative, which makes the reduction fail to work with arithmetic on numbers $\{0,1,\ldots,k,\infty\}$.

\bibliographystyle{alpha}
\bibliography{bib}


\end{document}